\documentclass[letterpaper,10 pt,conference]{ieeeconf}
\usepackage{preamble}

\title{\LARGE\bf  Drone Air Traffic Control:\\ Tracking a Set of Moving Objects with Minimal Power}
\author{
    Chek-Manh Loi$^{1,\dagger}$,
    Michael Perk$^{1,\dagger}$,
    Malte Hoffmann$^1$,
    Sándor Fekete$^{1,2}$
    \thanks{$^\dagger$ These authors share lead authorship}
\thanks{$^1$ Department of Computer Science, TU Braunschweig, Germany; \{loi,perk\}@ibr.cs.tu-bs.de, \{m.hoffmann,s.fekete\}@tu-bs.de; supported by the German Research Foundation (Deutsche Forschungsgemeinschaft, DFG) as part of project Computational Geometry: Solving Hard Optimization Problems (CG:SHOP) - 444569951.}
\thanks{$^2$ L3S, Germany}
}

\begin{document}
\maketitle
\begin{abstract}
    A common sensing problem is to use a set of stationary tracking locations
to monitor a collection of moving devices: Given $n$ objects that need to be
tracked, each following its own trajectory, and $m$ stationary traffic control
stations, each with a sensing region of adjustable range; how should we
adjust the individual sensor ranges in order to optimize energy consumption?
We provide both negative theoretical and positive practical results for this important and natural
challenge.
    On the theoretical side, we show that even if all objects move at constant speed along straight lines, no polynomial-time algorithm can guarantee optimal coverage for a given starting solution. 
    On the practical side, we present an algorithm based on 
    geometric insights that is able to find optimal solutions for the
\boldmath$\min \max$ variant of the problem, which aims at minimizing 
peak power consumption. 
Runtimes for instances with 500 moving objects and 25 stations are in the order of seconds 
for scenarios that take minutes to play out in the real world, demonstrating
real-time capability of our methods.
\end{abstract}

\section{Introduction}\label{sec:intro}
Keeping track of a (potentially large) collection of moving objects is a fundamental problem with a long history.
With the rise of air traffic, developing automated systems for this important task became even more critical: In \emph{air traffic control}, a typical scenario involves a set of moving planes, as well as a number of stationary control centers, which are equipped with powerful tracking devices such as long-range radar and various other communication and control capabilities to coordinate overall motion, without much regard for the power consumption of tracking devices.
With rapidly increasing number and ubiquity of small and light drones, this
type of air traffic control needs further
development~\cite{cohen2021urban,mazzenga20245g}, such as the use of less
powerful tracking stations: Given a collection of moving objects that must be
tracked by a set of stationary sensors with adjustable sensing regions,
how should we adjust the sensing radius over time to optimize power
consumption?

Formally, this is the Kinetic Disk Covering Problem (KDC), see \cref{fig:motivation} for an illustration.
Given 
a set $\pnts=\{p_1, \ldots, p_n\}$ of $n$ objects, 
each moving in space along a trajectory $p_i(t)$ over the time interval $t\in[0,1]$ as well as a set $\cnts$ of $m$ centers.
The goal is to assign a radius $r_i(t)$ to each center $\cnt_i \in \cnts$ at each time $t$ such that each point is covered by at least one disk, i.e., $\forall t \in [0,1], \forall p_j \in \pnts(t), \exists \cnt_i \in \cnts: ||\cnt_i - p_j(t)|| \leq r_i(t)$ and the maximum total area of all disks at any time is minimized, i.e., $\min \max_{t \in [0,1]} \sum_{i=1}^{m} \pi r_i(t)^2$.
Another variant of the KDC problem is to minimize the total area of all disks over time, i.e., $\min \int_0^1 \sum_{i=1}^{m} \pi r_i(t)^2 dt$.

The problem of finding an area-optimal assignment for a stationary set of points and centers is called the Disk Covering Problem (DC) and is known to be NP-hard~\cite{alt2006minimum}.
We show that the KDC problem is also NP-hard, even if we are given an optimal solution at time $t=0$.
This means that extending a solution over time in an optimal way is computationally infeasible.
We therefore study heuristic approaches for extending stationary solutions over time and provide an exact algorithm as well as efficient heuristics for the $\min \max$ variant of the KDC problem in \cref{sec:alg-approach}.
The proposed algorithms are then evaluated on several classes of instances of the KDC problem in \cref{sec:evaluation}. 
Despite the theoretic complexity, our results show that the proposed methods can find optimal solutions for realistically sized instances of the KDC problem in a matter of seconds.


\begin{figure*}[ht]
    \centering
    \includegraphics[]{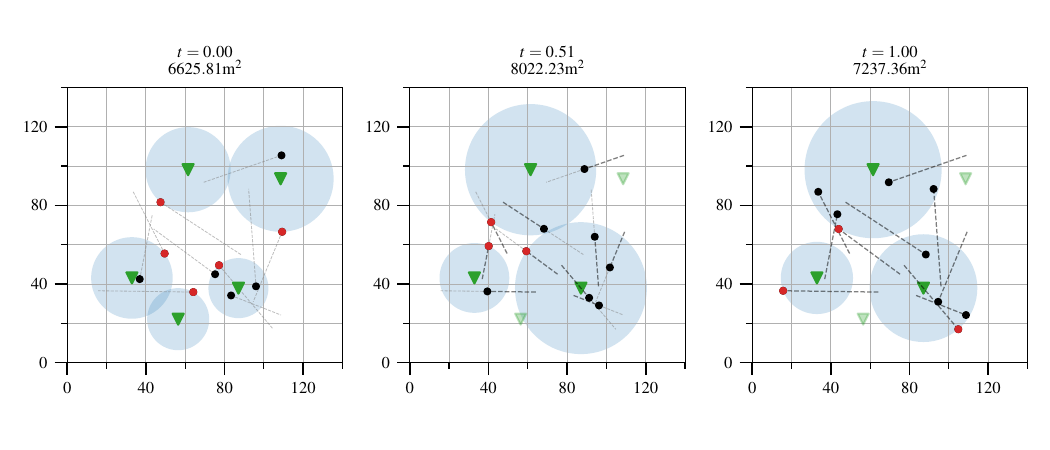}
    \vspace{-0.75cm}
    \caption{
        \label{fig:motivation}
        Solution of the KDC problem for $n=10$ moving points and $m=5$ stations.
        The points move along the dotted lines.
        The stations are shown as green triangles and are opaque, if the radius is zero.
        We start with an IP solution at time $\tm=0$.
        Our algorithm can find a solution over time, that uses at most \SI{8022.23}{\square\meter} of area over the whole time interval.
        This solution is matched by a lower-bound computed by solving the static DC problem at this time step.
    }
\end{figure*}


\section{Preliminaries}\label{sec:preliminaries}
We study the case in which each point $\pnts=\{p_1, \ldots, p_n\}$ moves along a given linear trajectory \(p_i(t) = (1-t) \cdot p_i^0 + t \cdot p_i^1 \in \mathbb{R}^2\) over the time interval $t\in[0,1]$.
We denote the set of points at time $\tm$ as $\pnts(\tm)$ and the stationary instance $\inst(\tm)=(\pnts(\tm), \cnts)$ of the DC at time $\tm$.
For simplicity, we sometimes use $\pnts$ to denote either a set of moving points or stationary points, with the intended meaning clarified by context.

In the KDC problem, we assign a radius $r_i(t)$ to each center $\cnt_i \in \cnts$ at each time $t$ such that all points are covered by at least one disk.
In an optimal solution, the radius of each disk is determined by the farthest point assigned to it.
We call these \suppnts{}.
A solution to the KDC is a list of $k$ assignments 
$\{\ass_{\tm_0},\ass_{\tm_1}, \ldots, \ass_{\tm_k}\},$ 
with $0=\tm_0\leq\tm_1\leq\ldots\leq\tm_k=1$.
Each assignment
$\ass_\tm : (\cnts\to\pnts)$ assigns \suppnts{} to centers over time.
The radius of the disk centered at $\cnt_i$ at time $\tm \in [\tm_j,\tm_{j+1})$ is $r_i(\tm) = \norm{\cnt_i - \ass_{\tm_j}(\cnt_i)(\tm)}$.

%


\section{Related Work}\label{sec:RelatedWork}
Given two point sets $\pnts$, $\cnts$ and a coverage function $\kappa$, the non-uniform minimum-cost multi-cover (MCMC) problem seeks disk radii so each $p \in \pnts$ is covered at least $\kappa(p)$ times. Abu-Affash~et~al.~\cite{abu2011multi} gave a $23.02 + 63.95(\kappa_{\text{max}} - 1)$-approximation for the uniform case, while\cite{bar2013note}~and~\cite{bhowmick2013constant} provided polynomial-time approximations for the non-uniform case. Huang~et~al.~\cite{huang2021ptas,huang2024ptas} presented the first PTAS for non-uniform MCMC, achieving solutions arbitrarily close to optimal. These works optimize disk costs using dynamic programming and recursive techniques.

Recently, there has been some work on multi-covering fixed points with disks of varying sizes~\cite{guitouni2025multi}.
Another line of closely related work examines multi-covering points with the minimum number of unit disks, where centers are arbitrary but radii are fixed. Gao~et~al.~\cite{gao2022fast} gave a 5-approximation algorithm running in $\bigO(n + \kappa_\text{max})$, and a 4-approximation with $\bigO(n^2)$ time. Filipov and Tomova~\cite{Georgiev2023} studied coverage with equal disks and proposed a stochastic optimization algorithm. These approaches focus on minimizing disk count rather than disk area.

Basch et al.~\cite{basch1999data} developed fundamental concepts for kinetic data structures.
Based on these results, Bespamyatnikh et al.~\cite{bespamyatnikh2000mobile} studied the \(k\)-center and \(k\)-means problem for both kinetic stations and objects.
Crevel et al.~\cite{crevel2012covering} studied a closely related problem of transmissioning messages to mobile receivers where the sensors are only activated at certain times to transmit a single message.

\section{Hardness}\label{sec:hardness}

Due to the known NP-hardness of the static Disk Covering (DC) problem~\cite{alt2006minimum}, the Kinetic Disk Covering (KDC) problem is NP-hard as well.
Suppose, however, that an optimal solution for the static instance at time $\tm_0$ is given. One may ask whether it is possible to optimally adapt this solution to a later time $\tm_1 > \tm_0$. The following lemma demonstrates that this is, in general, computationally intractable.

\begin{theorem}
    Given a KDC instance $I=(\pnts, \cnts)$, a cost bound $c$, and an optimal solution for the static instance at time $\tm_0$, determining whether there exists an assignment of radii at a later time $\tm_1 > \tm_0$ that covers all points with total cost at most $c$ is NP-hard.
\end{theorem}
\begin{proof}
    Given an instance $I_s(\pnts_s, \cnts_s)$ of the static DC problem, we can construct an instance $I_k(\pnts', \cnts')$ of the KDC problem as follows.
    For each static point $\pnt \in \pnts_s$ we add a point $\pnt'(\tm)=((1-\tm)\cdot (0,0)+ \tm \cdot \pnt)$ to $\pnts'$ and
    keep the centers, i.e., $\cnts'=\cnts_s$.
    
    Now at time $\tm_0=0$ the optimal solution is a single disk from the center closest to the origin $\cnt_0$ with radius $\norm{\cnt_0}$.
    A solution with cost at most $c$ at time $\tm_1=1$ exists if and only if a solution with cost at most $c$ exists for the static instance.
\end{proof}

\section{Stationary Disk Cover}\label{sec:static-dc}
In the stationary Disk Cover (DC) problem, we consider an instances of the KDC at a specific point $\tm$ in time, i.e., we consider the stationary instance $\inst(\tm) = (\pnts(\tm), \cnts)$.
We want to find an assignment of radii to the stations $\cnts$ such that each point in $\pnts(\tm)$ is covered by at least one disk and the total area of the disks is minimized.
For this, we provide a heuristic as well as an Integer Programming (IP) formulation that can solve the DC problem to provable optimally.

\subsection{Nearest-Neighbor Heuristic}
A simple heuristic to solve the DC problem is to assign each point to its nearest station.
In particular, we start by sorting all objects according to their distance to their nearest station in descending order and maintain a set of uncovered objects.
Then, we iterate over yet uncovered objects and assign each object to its nearest station in order to cover it with the smallest increase in area.
We then update the covered objects according to the new assignment and repeat until all objects are covered.

\subsection{Integer Programming}\label{sec:IP}
An effective method for finding optimal solutions to NP-hard problems is Integer Programming~(IP).
Although solving an IP can take exponential time in the worst-case scenario, using meticulously designed mathematical models, specialized algorithm engineering, and existing IP solvers allows for solving considerably large instances to provable optimality.
To formulate the DC as an IP, we define a discrete set of candidate disks $C$.

\subsubsection{Computing the Candidate Set}
\label{sec:candidate-disks}
The radius of each station is determined by its \suppnt.
This give rise to a remarkably simple enumeration scheme.
For each station, we consider every point $\pnt_i \in \pnts$ as a potential \suppnt{}.
Specifically, we first sort all points $\pnts$ in ascending order by their distance to the center $\cnt$.
Then, for each point $\pnt_i \in \pnts$, we add a disk centered at $\cnt$ with radius $\norm{\cnt - \pnt_i}$ to the candidate set $\candidateset$; this disk covers all points up to and including $\pnt_i$ in the ordering.

This yields $\bigO(m n)$ possible disks, as for each of the $m$ stations we consider $n$ potential support points.
Accounting for the sorting of the points for each station, the overall time to compute the candidate set is $\bigO(n m \log n)$.

\subsubsection{IP Formulation}\label{sec:ip-formulation}

For every disk $d_i$ in the candidate set $C$, we define a binary variable $x_i$ that encodes if the disk is used in the solution.
The constraint ensures that every point $\pnt_j\in \pnts$ is covered by at least one station.
\begin{equation*}
    \begin{array}{ll@{}ll}
        \text{minimize}  & \displaystyle \pi \cdot \sum\limits_{d_i\in C} &r_{i}^2x_{i} &\\
        \text{subject to}
                         & \displaystyle\sum_{\substack{d_i\in C\\ \pnt_j \in d_i}}   &x_{i} \geq 1,  & \forall \pnt_j \in \pnts\\
                         &&x_{i} \in \{0, 1\}, & \forall d_i \in C
    \end{array}
\end{equation*}

\section{Algorithmic Approaches}\label{sec:alg-approach}
Here we discuss algorithmic approaches to solve the KDC problem.
The idea is to extend every solution to the (stationary) DC problem at time $\tm$ to a solution for the KDC problem, see~\cref{sec:ensuring-feasibility}.
Further geometric insights allow us to improve extended solutions, see~\cref{sec:improving-extended-solutions}.
Finally, we can use these insights to design an iterative algorithm to solve the KDC problem, see~\cref{sec:sweep}.

\subsection{Ensuring Feasibility}\label{sec:ensuring-feasibility}
Every solution to the (stationary) DC problem at time $\tm$ can be extended to a solution for the KDC problem over time by maintaining the same assignment of objects to stations and altering the \suppnt{} of each station as needed to ensure coverage.
This yields a feasible solution for the KDC problem, but not necessarily an optimal one.
\cref{fig:ep1} shows, how initially the \suppnt{} of the station $\centername$ is $\pointbname$.
At time $\tm \approx 0.41$ the disks defined by $\pointaname$ and $\pointbname$ have the same radius.
After this point in time, the \suppnt{} of the disk changes to $\pointaname$.

\begin{figure}[htbp]
    \centering
    \begin{subfigure}[b]{.32\columnwidth}
        \centering
        \input{figures/tikz/ep1_1.tikz}
        \caption{
            $\tm = 0.2$
        }
    \end{subfigure}
    \centering
    \begin{subfigure}[b]{.32\columnwidth}
        \centering
        \input{figures/tikz/ep1_2.tikz}
        \caption{
            $\tm \approx 0.41$
        }
    \end{subfigure}
    \centering
    \begin{subfigure}[b]{.32\columnwidth}
        \centering
        \input{figures/tikz/ep1_3.tikz}
        \caption{
            $\tm = 0.6$
        }
    \end{subfigure}
    \caption{
        \label{fig:ep1}
        The \suppnt{} of station $\centername$ changes from $\pointaname$ to $\pointbname$, to ensure coverage of the objects at all times.
    }
\end{figure}

\begin{lemma}\label{lemma:suppnt-changes}
    In $\bigO(m n^2)$ time, all $\bigO(m n^2)$ possible \suppnt{} changes when extending a stationary solution to $t \in [0,1]$ can be computed.
\end{lemma}

\begin{proof}
    We can compute the time when the \suppnt{} changes, by solving the equation
    \begin{equation}
    \norm{\cnt - \pointaname(\tm)}^2 = \norm{\cnt - \pointbname(\tm)}^2.
    \end{equation}
    for two objects $\pointaname$ and $\pointbname$ assigned to the same station $\cnt$. 
    This yields a quadratic equation with at most two solutions,
    solvable analytically in $\bigO(1)$ time.

    As each station can have at most $n$ different \suppnts{} and each \suppnt{} can change to at most $n-1$ other objects, this yields $\bigO(m n^2)$ possible changes of \suppnts{}.
\end{proof}

Note that in degenerate cases the \suppnt{} might not change or multiple objects might have the same distance to a station, see \cref{fig:ep1_deg}.
We can resolve these cases, by choosing the object that is moving fastest away from the station as the new \suppnt{}.
This can be achieved by computing the derivative of the function defining the radius of $\cnt$ for each object $\pnt$ that has the same distance to the station at the event point.

\begin{figure}[htbp]
    \centering
    \begin{subfigure}[b]{.49\columnwidth}
        \centering
        \input{figures/tikz/ep1_deg_1.tikz}
    \end{subfigure}
    \centering
    \begin{subfigure}[b]{.49\columnwidth}
        \centering
        \input{figures/tikz/ep1_deg_2.tikz}
    \end{subfigure}
    \caption{
        Degenerate cases after which the assignment of \suppnt{} might change.
        In both cases we choose $\pointbname$ as the new \suppnt{}.
    }
    \label{fig:ep1_deg}
\end{figure}

\subsection{Improving Extended Solutions}\label{sec:improving-extended-solutions}
\Cref{sec:ensuring-feasibility} shows that stationary solutions can be extended by computing points in time where the \suppnt{} of a station changes.
Naturally, even when starting with an optimal solution for the stationary variant, this can lead to suboptimal solutions.
To reduce the overall cost, we consider a second type of event that changes the assignment of objects to stations, see \cref{fig:ep2}.
Intuitively, we want to identify points in time in which it is cheaper to handover a point from one station to another.

\begin{figure}[htbp]
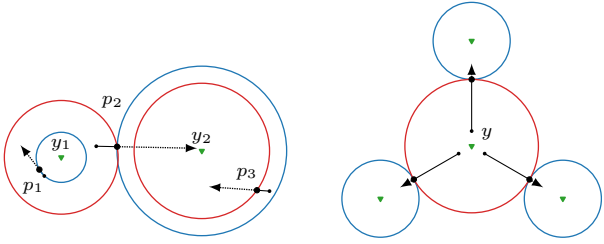

    \vspace{-0.5cm}
    \centering
    \begin{subfigure}[b]{.49\columnwidth}
        \centering
        \input{figures/tikz/ep2.tikz}
    \end{subfigure}
    \centering
    \begin{subfigure}[b]{.49\columnwidth}
        \centering
        \input{figures/tikz/ep3.tikz}
    \end{subfigure}
    \caption{
        (Left)
        Station $\cnt_2$ takes over a point $\pnt_2$.
        The objective values of the red and blue disks are equal at this point in time.
        We change the \suppnt{} of both $\cnt_1$ and $\cnt_2$ at the same time.
        (Right) 
        Generalization of the case depicted on the left to more than two stations.
    }
     \label{fig:ep2}
     \label{fig:ep3}
    \label{fig:EPS}
\end{figure}


\begin{lemma}\label{lemma:handover-changes}
    In $\bigO(m^2 n^3)$ time all $\bigO(m^2 n^3)$ possible handovers of objects between two stations (when extending a stationary solution to $t \in [0,1]$) can be computed.
\end{lemma}
\begin{proof}
    For two stations $\cnt_1$ and $\cnt_2$ with \suppnts{} $\pointaname$ and $\pointbname$, respectively, we consider the case in which $\cnt_2$ takes over the point $\pointbname$ from $\cnt_1$ at time $\tm$, see \cref{fig:ep2}.
    When $\cnt_2$ takes over a point from a station $\cnt_1$, the radius of the disk centered at $\cnt_1$ will shrink, as the \suppnt{} of $\cnt_1$ changes to the next furthest point $\pointcname$.
    At the same time the radius of the disk centered at the other station $\cnt_2$ will grow, as its \suppnt{} changes to the point it takes over from $\cnt_1$.
    Again, we can compute the time $\tm$ at which both solutions have the same cost, by solving a set of equations that can be solved analytically.
    \begin{align*}
        \norm{\cnt_1- \pointbname(\tm)}^2 &+ \norm{\cnt_2- \pointcname(\tm)}^2
        \\
        = \norm{\cnt_1- \pointaname(\tm)}^2 &+ \norm{\cnt_2- \pointbname(\tm)}^2
    \end{align*}

    As each station can have at most $n$ different \suppnts{} and each \suppnt{} can change to at most $n-1$ other objects, this yields $\bigO(m^2 n^3)$ possible handovers of objects between two stations.
\end{proof}


\begin{figure}[htbp]
    \centering
    \input{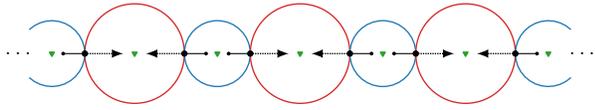}
    \caption{
        \label{fig:ep3_chain}
        Multiple stations and objects in the same event point.
        Either the blue or the red disk cover the objects, but not both.
    }
\end{figure}

The concept of handing over objects between stations can be generalized to three or more stations, see \cref{fig:ep3}. 
Furthermore, a chain of stations can successively take over objects from each other as time progresses, see~\cref{fig:ep3_chain}.
Initially, the optimal solution consists of blue disks covering all objects. 
As time progresses, there is a transition point at $\tm = 0.5$ for which the optimal solution switches to red disks, and the previously used blue disks are no longer part of the solution. 
This example can be extended to an arbitrarily long chain of objects and stations.
In the worst case, detecting these types of events requires comparing every subset of stations with every other subset, resulting in $\Omega(2^m)$ possible events.
Thus, we cannot hope to find an efficient algorithm that considers all possible handovers.

\subsection{Iterative Algorithm}\label{sec:sweep}
We present an algorithm that uses the observations from \cref{sec:ensuring-feasibility,sec:improving-extended-solutions} to compute heuristic and optimal solutions for the $\min \max$ KDC problem.
The algorithm maintains a solution over time that contains intervals in which the assignment of objects to stations does not change.

\subsubsection{Initial Solution}
The algorithm starts by computing a stationary solution for $\tm=0$ using any of the methods presented in \cref{sec:static-dc}.
In the case of the integer programming formulation, the algorithm can also provide a valid lower bound on the optimal solution for the $\min \max$ KDC problem which can be used to determine the quality of the solution during the solving process.

\subsubsection{Extending the Solution}\label{sec:alg-extending-solution}
Using a static solution at time $\tm$, we extend the solution over time by computing the necessary adjustments from \cref{lemma:suppnt-changes} to obtain a feasible solution.
To speed up the computation, we cache the next \suppnt{} change for each station $\cnt$ and update the cache whenever the assignment of $\cnt$ changes.

\subsubsection{Iterative Approach}
The algorithm then proceeds by identifying the point $\tm'\in [0,1]$ where the current solution has the maximum total area.
This point must lie at the intersection of two intervals in which the assignment of objects to stations changes because the objective value within each interval is an upward-opening parabola.
At this point in time we compute a new stationary solution for the DC problem using a method from \cref{sec:static-dc}.
If the stationary solution has equal or higher objective value than the current solution, we terminate as there is no hope of improving the current solution.
If the integer programming formulation is used, we can also terminate if the lower bound 
is within a desired optimality gap of the current solution. 
Initially, the (stationary) IP is solved only $1\%$ optimality,
which is then gradually decreased to $0.01\%$ when the current KDC solution is within $1.5\%$ of the lower bound.

If the stationary solution has a smaller objective value than the current solution, we can improve our current solution at time $\tm'$ and potentially over the whole interval $[0,1]$ or any set of subintervals.
This is done by extending the solution at time $\tm'$ with the methods from \cref{sec:alg-extending-solution} both in the positive and negative time direction.
Afterwards the solutions are merged into a single solution that is valid over the entire time interval, see \cref{sec:alg-combining-previous-solutions}.
This process is repeated until no further improvement is possible.

\subsubsection{Combining with Previous Solutions}\label{sec:alg-combining-previous-solutions}
Given a new solution for an interval $[\tm_1, \tm_2]$, we can combine it with the current best solution.
This can be done by computing the intersection points of the objective values of the two solutions.
As each solution consists of intervals in which the assignment of objects to stations does not change, the objective value of each solution is piece-wise quadratic and thus this intersection can be computed efficiently.
We then compute the lower envelope of the two solutions in $\bigO(k_1 + k_2)$ time, where $k_1$ and $k_2$ are the number of intervals in the two solutions.
This yields a new solution that is valid over the entire time%
, see \cref{fig:sweep}.

\begin{figure}[t]
    \centering
    \input{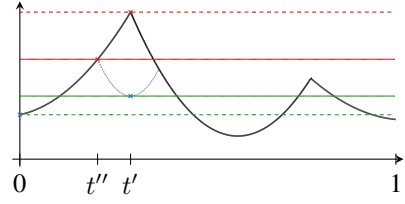}
    \caption{
        Iterative approach to improve the $\min\max$ power consumption over time.
        Green and red lines denote lower and upper bounds on the optimal solution. 
    } \label{fig:sweep}
    \vspace{-0.5cm}
\end{figure}

\subsubsection*{Termination}
The algorithm terminates when no further improvement is possible, or when the lower bound is within a desired optimality gap of the current solution.
We can show that the algorithm always terminates.
\begin{lemma}
    The algorithm terminates after a finite number of iterations.
\end{lemma}

\newcommand{\lowerenvelope}{\mathcal{L}}
\begin{proof}
    There are only finitely many assignments of objects to stations.
    For any fixed assignment, the objective value as a function of time is a quadratic polynomial.
    However, the assignment may only be feasible on a subset of the time interval.
    By \cref{lemma:suppnt-changes}, feasibility can only change finitely many times. Thus, each quadratic function can be partitioned into a finite number of feasible pieces.
    Consequently, there are only finitely many assignment–interval pieces to consider across all assignments.

    The lower envelope \(\lowerenvelope\) of these pieces yields the optimal solution that minimizes the total area over time, i.e., it is optimal at every $t\in[0,1]$.
    Whenever we compute an optimal static solution using our integer program, our combined solution includes a piece of \(\lowerenvelope\).
    Once our algorithm computes an optimal solution on the same piece of \(\lowerenvelope\) twice, i.e., it first proves a lower bound and later finds the maximum in the same piece, no better solution can be found and the algorithm terminates.
    Because there are only finitely many pieces, the algorithm must terminate after finitely many iterations.
\end{proof}

\subsubsection*{Improvements to the Algorithm}
We propose further improvements to the above algorithm.
The first improvement is to the extend solution routine from \cref{sec:alg-extending-solution} by considering handovers of objects between two stations, see \cref{lemma:handover-changes}. Again an efficient caching scheme can be used to speed up the computation.

The second improvement targets assignments that are obtained while extending the solution.
Given a suboptimal solution for the DC problem, we can improve it by altering the assignment of objects to stations.
If the \suppnt{} of a center $\cnt_i$ is also covered by another disk $\cnt_j$, we assign the second furthest object of $\cnt_i$ as the new \suppnt{} of $\cnt_i$.
This reduces the radius of $\cnt_i$ and therefore the total area of the disks.
We can repeat this process until no further improvement is possible.

Finally, we reduce the time spent for extending and combining solutions. 
This can be done by simultaneously extending and combining the stationary solution (which is better than the current solution) only until it intersects the current solution.
The influence of both improvement strategies is evaluated in \cref{sec:evaluation}.

\newcommand{\fixedn}{\texttt{fix\_n}}
\newcommand{\fixedm}{\texttt{fix\_m}}
\newcommand{\fixed}{\texttt{fix}}
\newcommand{\public}{\texttt{pub}}

\newcommand{\algIP}{\texttt{IP}}
\newcommand{\algNN}{\texttt{NN}}
\newcommand{\algFixedNN}{\texttt{FixedNN}}

\newcommand{\stratRemoveDup}{\texttt{NoDup}}
\newcommand{\stratImprovedExtension}{\texttt{ImpExt}}
\newcommand{\stratPartialExtend}{\texttt{PartExt}}

\section{Evaluation}\label{sec:evaluation}

Experiments were carried out on Linux desktop workstations with
AMD Ryzen 9 7900 ($12\times$\SI{3.7}{\giga\hertz}) CPUs and \SI{96}{\giga\byte} of RAM running Ubuntu 24.04.3 LTS.
Code and data with benchmark problems and results are available~\cite{ouranonymousgit}.
We use the Gurobi~\cite{gurobi} solver for the integer program, with a default optimality gap of $0.01\%$ and a time limit of \SI{600}{\second}. 
To detect numerical issues for degenerate cases, we added checks to ensure correctness of our implementation.
A version that uses exact number types for finding roots of the quadratic equations was also tested, see \cref{sec:evaluation-instances-from-literature} for details.
We formulate the following research questions:
\begin{enumerate}
    \item[RQ1] Which of the proposed improvement strategies is effective in practice?
    \item[RQ2] How does the heuristic and IP perform for different values of $m$ and $n$?
    \item[RQ3] How does the difficulty of instances change when considering degenerate cases?
    \item[RQ4] How does the algorithm perform on instances from literature?
\end{enumerate}

For the instance generation, we used line segments with lengths uniformly distributed in $[25,50]$ in a $100\times 100$ grid.
In real-world applications where objects are drones have velocities up to \SI{100}{\kilo\metre\per\hour}, this resembles scenarios that take between $15$ and $30$ minutes to play out on a \SI{100}{\kilo\metre\squared} canvas.
In total, we have the following instances sets.
\begin{description}
    \item[\fixed]
        We generate 25 instances with $m=25$ and $n=500$.
    \item[\fixedn]
        We fix the number of objects to $n=500$.
        For each number of sensors $m\in\{5,10,\ldots,50\}$ we generate $10$ instances.
    \item[\fixedm]
        We fix the number of sensors to $m=25$.
        For each number of objects $n\in\{50,100,\ldots,500\}$ we generate $10$ instances.
    \item[\public]
        Instances from well-known publicly available benchmarks used in~\cite{fekete_et_al:LIPIcs.SoCG.2025.48}.
      These include point sets from the CG:SHOP challenges~\cite{demaine2022area,demaine2020computing},
      TSPLIB instances~\cite{reinelt1991}, instances from a VLSI dataset~\cite{VLSI} and point sets from the Salzburg Database of Polygonal Inputs~\cite{EDER2020105984}.
      We scaled point sets to \SI{1}{\metre} resolution, such that they have a diameter of \SI{100}{\kilo\metre} and sampled $m=25$ centers uniformly at random.
      The minimum and maximum length of trajectories was set to \SI{25}{\kilo\metre} and \SI{50}{\kilo\metre}. 
      This yields a graph $G$ with all possible trajectories as edges.
      We compute a random maximum cardinality matching in $G$, yielding $\lfloor \frac{N-25}{2}\rfloor$ objects in almost all cases.
      As this yielded some degenerate cases in which many trajectories had the same length, we added a small random perturbation to the trajectories to break ties.
\end{description}

\begin{figure}[ht]
    \centering
    \includegraphics[]{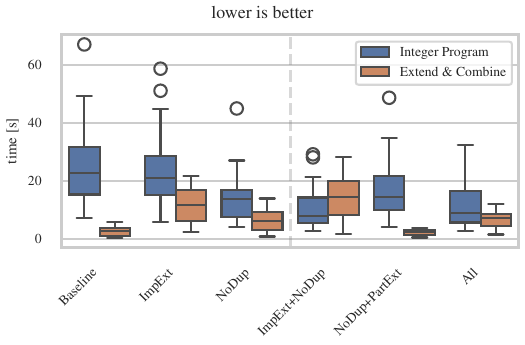}
    \caption{
        Evaluation of different improvement strategies on the \fixed{} benchmark set.
        The results show that enabling all three solution strategies yields the best performance.
    }\label{fig:configurations}
\end{figure}

\subsection{RQ1: Improvement Strategies}
\label{sec:evaluation-improvement-strategies}
We evaluate the influence of three proposed improvement strategies to the algorithm from \cref{sec:alg-approach}. 
We denote the removal of duplicates as {\stratRemoveDup}, the improved extension according to handovers from \cref{lemma:handover-changes} as {\stratImprovedExtension} and partial extension as {\stratPartialExtend}.
Our baseline is the exact algorithm that uses the IP solver from \cref{sec:IP}.
We compare the total runtime of the different components of the algorithm, see~\cref{fig:configurations}, i.e., the time to compute the stationary IP solutions and the time to extend and combine solutions.
We see that all improvement strategies have a positive effect on the IP solving time.
Our experiments show however, that both \stratImprovedExtension{} and \stratRemoveDup{} have a negative effect on the time to extend and combine solutions which in case of \stratImprovedExtension{} results in a higher overall runtime compared to the baseline.

The use of \stratPartialExtend{} however reduces the time spent during extension and combination significantly, which mitigates the negative effect of the other two strategies while still reducing the time spent in the IP solver.
This results in the best configuration that enables all three improvement strategies and achieves an overall runtime reduction of about $39\%$ compared to the baseline.

\subsection{RQ2: Influence of $m$ and $n$}
\label{sec:evaluations-influence-m-n}
We evaluate the influence of the number of objects~$n$ and the number of stations~$m$ on the performance of the algorithm.
Our comparison uses the best configuration from \cref{sec:evaluation-improvement-strategies}.
We execute the algorithm from \cref{sec:alg-approach} with both an IP solver and the nearest neighbor heuristic from \cref{sec:static-dc} to solve the DC problem in each iteration.
We denote the IP-based algorithm as {\algIP} and the nearest neighbor based algorithm as {\algNN}.
Additionally, we implemented another heuristic that divides the time $[0,1]$
into $k=10$ evenly spaced intervals and computes the nearest neighbor solution
at each border. %
The algorithm \algFixedNN{} then extends the solutions using the routine from \cref{sec:alg-extending-solution} and
reports the lower envelope of all computed solutions.
We note that one instance of \algIP{} on the \fixedn{} benchmark was restarted with a different seed due to numerical instability in the non-exact extension step that caused our checks to fail.

\begin{figure}[h]
    \includegraphics[trim={0 0.1cm 0 0},clip]{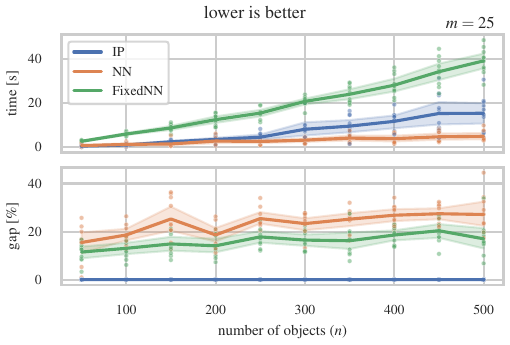}
    \includegraphics[trim={0 0 0 0.29cm},clip]{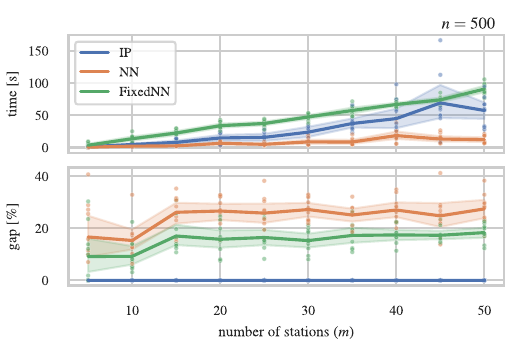}
    \caption{
        Performance of \algIP{}, \algNN{} and \algFixedNN{} on instances from the \fixedm{} and \fixedn{} datasets.
        While \algIP{} is significantly slower than \algNN{}, it produces optimal solutions for all instances.
        \algFixedNN{} is slightly better in terms of solution quality than \algNN{} but significantly slower.
    }
    \label{fig:fixednm}
\end{figure}

We evaluate the performance of the different algorithms on \fixedm{} and \fixedn{} to access the influence of $m$ and $n$ on the overall runtime and quality of solutions, see \cref{fig:fixednm}.
For the \fixedm{} benchmark, we see that \algIP{} produces optimal solutions for instances with up to $500$ objects in less than \SI{30}{\second}.
The \algNN{} heuristic is significantly faster, solving all instances within $6$ seconds.
Although \algNN{} is faster, it produces solutions with an average optimality gap of about $20\%$ for smaller instances and up to $45\%$ for larger instances.
The \algFixedNN{} is slightly better in terms of solution quality than \algNN{} with an average optimality gap of around $20\%$ for all instances but is significantly slower with an average runtime even higher than \algIP.
For the \fixedn{} benchmark, we see a similar trend but observe a higher influence of $m$ on the performance of the algorithms.
In particular doubling the number of stations from $m=25$ to $50$ for $n=500$ objects more than doubles the maximum runtime of \algIP{} from about \SI{25}{\second} to about \SI{70}{\second}.

In summary, \algIP{} is able to solve all instances to optimality in a matter of seconds, demonstrating real-time capability for scenarios that take minutes to play out in the real world.

\subsection{RQ3: Degenerate Cases}
\label{sec:evaluation-degenerate-cases}

In many real-world scenarios, objects may move in a coordinated manner.
Planes often follow established air corridors, and drones may operate in formations or along predefined paths.
We evaluate the performance of the best algorithm \algIP{} from \cref{sec:evaluations-influence-m-n} on on four classes of degenerate instances;
one with uniformly random trajectories (\fixed{}), one where all have the same slope, and two where all have the same start or end point, see \cref{fig:degeneracies}.

We observe that instances with same slope are slightly easier to solve.
On the other hand instances where all trajectories start or end at the same point are significantly easier to solve with all instances being solved within seconds.

\begin{figure}[h]
    \centering
    \includegraphics[]{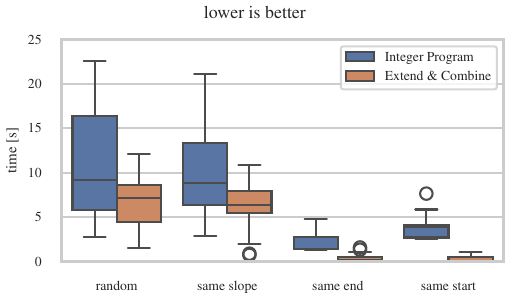}
    \caption{
        Runtime comparison between degenerate cases: uniformly random trajectories, same slope, same start point, and same end point.
        Degenerate instances are consistently easier to solve than random instances, suggesting that real-world instances may be more tractable.
        Note that this plot is clipped; one random and one same-slope instance of the integer program (with runtimes of \SI{33}{\second} and \SI{38}{\second})  exceed the displayed range.
    }\label{fig:degeneracies}
\end{figure}

\subsection{RQ4: Instances from Literature}
\label{sec:evaluation-instances-from-literature}
We further assess the algorithm's performance on the \public{} instance set, as shown in \cref{fig:public} and illustrated by \cref{fig:instance_example}. 
The presence of degeneracies originating from real-world data, led to various numerical challenges in our implementation, see \cref{sec:evaluations-influence-m-n}.
To address these issues, we developed an alternative version of \algIP{} utilizing exact number types for computations, following \cref{lemma:suppnt-changes,lemma:handover-changes}. 
The additional overhead of exact number types results in slightly higher runtimes than in the benchmarks reported in \cref{sec:evaluations-influence-m-n}.
However, our results show that $290$ out of $302$ instances with up to $700$ objects could be solved to provable optimality within \SI{600}{\second} of computation time. 

\begin{figure}
    \vspace{-3mm}
    \centering
    \includegraphics[]{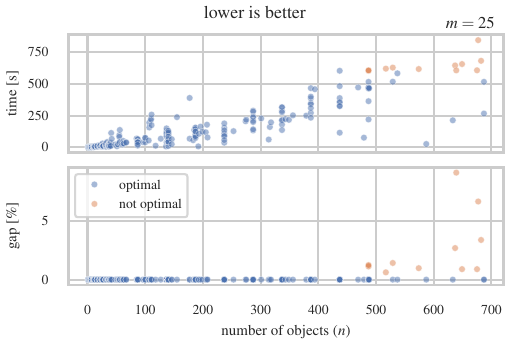}
    \caption{
        Performance of the exact version of \algIP{} on the \public{} dataset. Despite the variety and degeneracies, almost all instances can be solved to provable optimality.
    }\label{fig:public}
    \vspace{-3mm}
\end{figure}

\subsection{Real-World Perspectives}

While our results establish both hardness and optimal solvability for the $\min\max$ variant of the KDC under idealized assumptions, several additional
factors need to be addressed for real-world scenarios.
Overall, our results establish theoretical limits and algorithmic potential under idealized geometric assumptions. Practical deployment requires 
additional considerations for practical issues such as communication delays, uncertainty models, and possibly online adaptation mechanisms.

\paragraph{Modeling Abstraction}
Our formulation is purely geometric and assumes deterministic trajectories known in advance.
We do not model sensor dynamics (e.g., delays when adjusting ranges), communication constraints, or decentralized coordination. 
In practical UAV or air-traffic systems, range updates may incur latency, bandwidth limitations, or partial communication failures.
Such effects may delay handovers or introduce temporary coverage gaps, potentially increasing peak power consumption relative to the idealized optimum.
In principle, our algorithm remains applicable under moderate latency if event times are buffered with safety margins; however, 
practical performance guarantees may degrade accordingly.

\paragraph{Real-Time Adaptivity}
Although polynomial, the event-based enumeration is evaluated only in an offline setting with pre-known trajectories.
In adaptive scenarios with trajectories revealed incrementally or updated online, the recomputation cost may limit scalability.
This is where incremental variants are still necessary for large-scale deployments, at the expense of global optimality guarantees.

\paragraph{Dimensionality and Physical Effects}
Our work focuses on 2D planar coverage and ideal sensing disks to demonstrate basic viability.
Real-world systems operate in 3D, for which additional effects like altitude and occlusion may affect coverage.
Extending the geometric insights to 3D is possible; practical application needs to 
address increased computational complexity and further structural properties of optimal solutions.

\paragraph{Randomness and Uncertainty}
Our experiments demonstrate possibilities for relatively simple trajectories.
They do not explicitly model stochastic motion, sensing noise, or localization
error. In practice, uncertainty in object position or velocity needs to address
robustness margins, effectively enlarging sensing regions and reducing
achievable power savings.

\section{Conclusions}\label{sec:conclusion}
We provide novel insights into tracking large sets of moving objects from a set of fixed observation
stations with optimal energy.
While this problem is hard in theory,
our practical methods can compute provably 
optimal coverage in a matter of seconds, demonstrating perspectives for real-world capability.
There are many promising directions for future work;
deployment requires addressing the described practical issues, such as communication delays, online trajectory updates, and three-dimensional coverage.
If the stations can be repositioned, the problem becomes even more complex, but also opens up new possibilities for optimization that may further reduce energy consumption.
Uncertainty in trajectories and alternative coverage models, such as probabilistic sensing,
also present interesting avenues for further research that could enhance the robustness and applicability of our methods in real-world scenarios.


\begin{figure}[h]
    \includegraphics[]{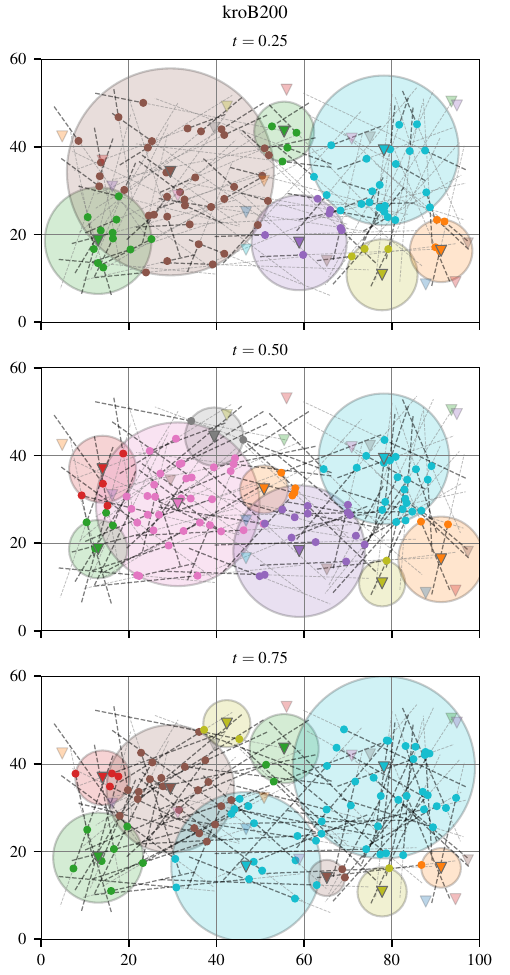}
    \caption{
        \label{fig:instance_example}
        A solution of an instance from \public{} instance set based on the \emph{kroB200} TSPLIB~\cite{reinelt1991} point set at $\tm=0.25,\tm=0.5$ and $\tm=0.75$.
    }
\end{figure}

\ 
\newpage

\bibliographystyle{IEEEtranS}
\bibliography{main.bib}

@article{abu2011multi,
  title     = {Multi cover of a polygon minimizing the sum of areas},
  author    = {Abu-Affash, A Karim and Carmi, Paz and Katz, Matthew J and Morgenstern, Gila},
  journal   = {International Journal of Computational Geometry \& Applications},
  volume    = {21},
  number    = {06},
  pages     = {685--698},
  year      = {2011},
  publisher = {World Scientific}
}

@article{bar2013note,
  title     = {A note on multicovering with disks},
  author    = {Bar-Yehuda, Reuven and Rawitz, Dror},
  journal   = {Computational Geometry},
  volume    = {46},
  number    = {3},
  pages     = {394--399},
  year      = {2013},
  publisher = {Elsevier}
}

@inproceedings{bhowmick2013constant,
  title     = {A constant-factor approximation for multi-covering with disks},
  author    = {Bhowmick, Santanu and Varadarajan, Kasturi and Xue, Shi-Ke},
  booktitle = {Symposium on Computational Geometry (SoCG)},
  pages     = {243--248},
  year      = {2013}
}

@inproceedings{huang2021ptas,
  title     = {{PTAS} for minimum cost multi-covering with disks},
  author    = {Huang, Ziyun and Feng, Qilong and Wang, Jianxin and Xu, Jinhui},
  booktitle = {Symposium on Discrete Algorithms (SODA)},
  pages     = {840--859},
  year      = {2021}
}

@article{huang2024ptas,
  title   = {{PTAS} for Minimum Cost MultiCovering with Disks},
  author  = {Huang, Ziyun and Feng, Qilong and Wang, Jianxin and Xu, Jinhui},
  journal = {SIAM Journal on Computing},
  volume  = {53},
  number  = {4},
  pages   = {1181--1215},
  year    = {2024}
}

@inproceedings{gao2022fast,
  title     = {Fast approximation algorithms for multiple coverage with unit disks},
  author    = {Gao, Xuening and Guo, Longkun and Liao, Kewen},
  booktitle = {Symposium on a World of Wireless, Mobile and Multimedia Networks (WoWMoM)},
  pages     = {185--193},
  year      = {2022}
}

@inproceedings{Georgiev2023,
  author    = {Filipov, Stefan M.
               and Tomova, Fani N.},
  editor    = {Georgiev, Ivan
               and Datcheva, Maria
               and Georgiev, Krassimir
               and Nikolov, Geno},
  title     = {Covering a Set of Points with a Minimum Number of Equal Disks via Simulated Annealing},
  booktitle = {Numerical Methods and Applications},
  year      = {2023},
  pages     = {134--145}
}

@inproceedings{guitouni2025multi,
  author    = {Mariem Guitouni and
               Chek{-}Manh Loi and
               S{\'{a}}ndor P. Fekete and
               Michael Perk and
               Aaron T. Becker},
  title     = {Multi-Covering a Point Set by m Disks with
               Minimum Total Area},
  booktitle = {ICRA},
  pages     = {3000--3006},
  publisher = {{IEEE}},
  year      = {2025},
  doi       = {10.1109/ICRA55743.2025.11127835},
  bibsource = {dblp computer science bibliography, https://dblp.org}
}

@inproceedings{alt2006minimum,
  author    = {Helmut Alt and
               Esther M. Arkin and
               Herv{\'{e}} Br{\"{o}}nnimann and
               Jeff Erickson and
               S{\'{a}}ndor P. Fekete and
               Christian Knauer and
               Jonathan Lenchner and
               Joseph S. B. Mitchell and
               Kim Whittlesey},
  title     = {Minimum-cost coverage of point sets by disks},
  booktitle = {Symposium on Computational Geometry (SoCG)},
  pages     = {449--458},
  year      = {2006},
  _url      = {https://doi.org/10.1145/1137856.1137922},
  doi       = {10.1145/1137856.1137922}
}

@misc{gurobi,
  author = {{Gurobi Optimization, LLC}},
  title  = {{Gurobi Optimizer Reference Manual}},
  year   = 2025,
  url    = {https://www.gurobi.com}
}

@misc{VLSI,
  title = {{VLSI} Data Set},
  author = {{Rohe, Andre}},
  url   = {https://www.math.uwaterloo.ca/tsp/vlsi/index.html}
}

@article{mazzenga20245g,
  author   = {Mazzenga, Franco and Giuliano, Romeo and Vizzarri, Alessandro},
  journal  = {IEEE Access},
  title    = {5G-Based Synchronous Network for Air Traffic Monitoring in Urban Air Mobility},
  year     = {2024},
  volume   = {12},
  number   = {},
  pages    = {188542-188559},
  doi      = {10.1109/ACCESS.2024.3513212},
  issn     = {2169-3536},
  month    = {}
}

@article{cohen2021urban,
  author   = {Cohen, Adam P. and Shaheen, Susan A. and Farrar, Emily M.},
  journal  = {IEEE Transactions on Intelligent Transportation Systems},
  title    = {Urban Air Mobility: History, Ecosystem, Market Potential, and Challenges},
  year     = {2021},
  volume   = {22},
  number   = {9},
  pages    = {6074-6087},
  doi      = {10.1109/TITS.2021.3082767}
}

@article{crevel2012covering,
  author    = {Crevel Bautista{-}Santiago and
               Jos{\'{e}} Miguel D{\'{\i}}az{-}B{\'{a}}{\~{n}}ez and
               Ruy Fabila Monroy and
               David Flores{-}Pe{\~{n}}aloza and
               Dolores Lara and
               Jorge Urrutia},
  title     = {Covering moving points with anchored disks},
  journal   = {Eur. J. Oper. Res.},
  volume    = {216},
  number    = {2},
  pages     = {278--285},
  year      = {2012},
  _url      = {https://doi.org/10.1016/j.ejor.2011.07.048},
  doi       = {10.1016/J.EJOR.2011.07.048},
  timestamp = {Mon, 03 Mar 2025 21:38:02 +0100},
  biburl    = {https://dblp.org/rec/journals/eor/Bautista-SantiagoDMFLU12.bib},
  bibsource = {dblp computer science bibliography, https://dblp.org}
}

@article{basch1999data,
  author    = {Julien Basch and
               Leonidas J. Guibas and
               John Hershberger},
  title     = {Data Structures for Mobile Data},
  journal   = {J. Algorithms},
  volume    = {31},
  number    = {1},
  pages     = {1--28},
  year      = {1999},
  _url      = {https://doi.org/10.1006/jagm.1998.0988},
  doi       = {10.1006/JAGM.1998.0988},
  timestamp = {Sat, 08 Jun 2024 13:16:29 +0200},
  biburl    = {https://dblp.org/rec/journals/jal/BaschGH99.bib},
  bibsource = {dblp computer science bibliography, https://dblp.org}
}

@inproceedings{bespamyatnikh2000mobile,
  author    = {Sergei Bespamyatnikh and
               Binay K. Bhattacharya and
               David G. Kirkpatrick and
               Michael Segal},
  title     = {Mobile facility location},
  booktitle = {International Workshop on Discrete Algorithms
               and Methods for Mobile Computing and Communications {(DIAL-M)}},
  pages     = {46--53},
  year      = {2000},
  _url      = {https://doi.org/10.1145/345848.345858},
  doi       = {10.1145/345848.345858},
  timestamp = {Tue, 06 Nov 2018 16:58:28 +0100},
  biburl    = {https://dblp.org/rec/conf/dialm/BespamyatnikhBKS00.bib},
  bibsource = {dblp computer science bibliography, https://dblp.org}
}

@article{demaine2020computing,
  title   = {Computing convex partitions for point sets in the plane: The {CG:SHOP Challenge} 2020},
  author  = {Demaine, Erik D and Fekete, S{\'a}ndor P and Keldenich, Phillip and Krupke, Dominik and Mitchell, Joseph S. B.},
  journal = {arXiv preprint arXiv:2004.04207},
  year    = {2020}
}

@article{demaine2022area,
  title     = {Area-optimal simple polygonalizations: The {CG Challenge} 2019},
  author    = {Demaine, Erik D and Fekete, S{\'a}ndor P. and Keldenich, Phillip and Krupke, Dominik and Mitchell, Joseph S. B.},
  journal   = {Journal of Experimental Algorithmics (JEA)},
  volume    = {27},
  number    = {2},
  pages     = {1--12},
  year      = {2022},
  publisher = {ACM New York, NY}
}

@article{reinelt1991,
  author    = {Reinelt, G.},
  journal   = {ORSA Journal of Computing},
  number    = {4},
  pages     = {376--384},
  title     = {{{TSPLIB}--A Traveling Salesman Problem Library}},
  volume    = {3},
  year      = {1991}
}

@article{EDER2020105984,
  author   = {G{\"u}nther Eder and Martin Held and Stein{\th}{\'{o}}r Jasonarson and Philipp Mayer and Peter Palfrader},
  doi      = {10.1016/j.dib.2020.105984},
  issn     = {2352-3409},
  journal  = {Data in Brief},
  pages    = {105984},
  title    = {Salzburg Database of Polygonal Data: Polygons and Their Generators},
  volume   = {31},
  year     = {2020}
}

@inproceedings{fekete_et_al:LIPIcs.SoCG.2025.48,
  author    = {Fekete, S\'{a}ndor P. and Keldenich, Phillip and Perk, Michael},
  title     = {{Exact Algorithms for Minimum Dilation Triangulation}},
  booktitle = {Symposium on Computational Geometry (SoCG)},
  pages     = {48:1--48:18},
  year      = {2025},
  doi       = {10.4230/LIPIcs.SoCG.2025.48}
}

@misc{ouranonymousgit,
  author    = {Loi, Chek-Manh and Perk, Michael},
  title     = {Drone Air Traffic Control: Tracking a Set of
               Moving Objects with Minimal Power
               },
  year      = 2025,
  publisher = {Zenodo},
  doi       = {10.5281/zenodo.17119781},
  url       = {https://doi.org/10.5281/zenodo.17119781}
}

\end{document}